\documentclass[12pt]{iopart}
\usepackage{iopams}
\usepackage{dsfont}
\usepackage{mathrsfs}
\usepackage{graphicx,color}

\usepackage{hyperref}

\hypersetup{bookmarks=true,
colorlinks=true,
linkcolor=blue, citecolor=blue, menucolor=blue, filecolor=blue, urlcolor=blue,
 pdfauthor={Olaf Milbredt}, pdftitle={The Cauchy Problem for
 Membranes}, pdfsubject={Geometric Analysis},
pdfkeywords={Nonlinear Wave Equations, Membranes, Cauchy Problem}}

\newcounter{vol2}
\newcommand {\rr}{{\mathds R}}
\newcommand {\rrangle}{{\rangle\kern - 2.5pt \rangle}}
\newcommand {\llangle}{{\langle\kern - 2.5pt \langle}}

\newcommand{\restr}[1]{#1\bigr|_{t=0}}

\newcommand{\abs}[1]{|#1|}
\newcommand{\babs}[1]{\bigl|#1\bigr|}

\newcommand {\init}[1]{{\mathring{#1}}}
\newcommand{\II}{{\init{I\kern - 2.0pt I}}{}}
\newcommand{\IIfull}{{I\kern - 2.0pt I}{}}
\newcommand{\slice}{\mathscr{S}}
\newcommand{\D}{\mathbf{D}}

\newcommand{\G}{\mathbf{\Gamma}}

\newcommand{\hT}{\widehat{T}}

\newcommand{\hn}{\widehat{\nabla}}

\newcommand{\ig}{\init{g}{}}

\newcommand{\inab}{\init{\nabla}{}}

\newcommand{\RRiem}{\mathbf{Rm}}

\newcommand{\im}{\mathrm{im}}
\newcommand{\dom}{\mathrm{domain}}

\newcommand{\id}{\mathrm{id}}

\newenvironment{proof}
{\noindent\textbf{Proof:}\\ } {\hfill \opensquare \\}

    \newtheorem{thm}{Theorem}[section]
    \newtheorem{lem}[thm]{Lemma}
    \newtheorem{prop}[thm]{Proposition}
    \newtheorem{cor}[thm]{Corollary}

    \newtheorem{defn}[thm]{Definition}
    \newtheorem{rem}[thm]{Remark}

\begin{document}
\title{The Cauchy Problem for Membranes}
\author{Olaf Milbredt} 
\address{
 Max Planck Institute for Gravitational Physics, 
Am M\"uhlenberg 1,
14476 Potsdam, Germany
 }
\ead{Olaf.Milbredt@aei.mpg.de}

\begin{abstract}
  We show existence and uniqueness for timelike minimal submanifolds
  (world volume of $p$-branes)
  in ambient Lorentz manifolds admitting a time function in a neighborhood
  of the initial submanifold. The initial value
  formulation introduced and the conditions imposed on the initial data
  are given in purely geometric terms.
  Only an initial direction must be prescribed
  in order
  to provide uniqueness for the geometric problem.
  The result covers non-compact initial submanifolds of any codimension.
  By considering the angle of the initial direction and vector fields
  normal to the initial submanifold with the unit normal to the
  foliation given by the time function we obtain a quantitative description
  of ``distance'' to the light cone. This description makes it possible
  to treat initial data which are arbitrarily close to the light cone.
  Imposing uniform assumptions give a lower bound for a notion of
  ``time of existence'' depending only on geometric quantities involving
  the length of timelike curves lying in the solution.
\end{abstract}

\submitto{\CQG}

\maketitle

\section{Introduction}
Membranes are spacelike submanifolds evolving in a Lorentzian manifold.
The equation of motion of a membrane is determined by 
the condition that the world volume of the membrane, the timelike submanifold 
swept out
during  time evolution, is a critical
point for the volume functional induced by the ambient manifold.
Membranes arise in the context of higher-dimensional extensions of
String Theory, where 
they are called $p$-branes according to the dimension of the spacelike object.
The Euler-Lagrange equation of the volume functional is the vanishing of the 
mean extrinsic curvature vector %
of the 
world volume. %
This equation will be called the membrane equation in
the sequel.
Due to the signature of the world volume the membrane  equation
leads to a system of nonlinear wave equations.
Since the membrane equation %
 is invariant under 
diffeomorphisms of the world volume, it %
is degenerate.

In this paper we consider the solvability of the 
initial value problem (IVP) %
for the membrane equation with %
initial data %
only involving geometric
quantities. Our focus is on providing short-time existence and uniqueness
for the Cauchy problem admitting non-compact initial data, and data which are 
arbitrarily close to the light cone.
The Cauchy problem for closed strings, i.\,e.\ a membrane
diffeomorphic to the circle line, has been studied 
in Minkowski space and in globally hyperbolic ambient manifolds
by
T.\ Deck \cite{Deck:1994}, %
and O.\ M\"uller \cite{Muller:2007}. %
Solutions are obtained as timelike immersions in a neighborhood of the initial
data. In these works the question of global-in-time existence 
was investigated for globally hyperbolic ambient manifolds. 
Exact solutions of this equation
corresponding to pulsating and rotating objects
have been studied by H.\ Nicolai and J.\ Hoppe in \cite{NiHo:1987}.

In the case of the world volume in Minkowski space represented as a graph,
global existence for small
initial data was shown in \cite{Lind:2004}. 
This result was generalized to arbitrary codimensions in \cite{AA:2006}.

Using techniques previously applied to the Einstein equation
by D.\ 
Christodoulou %
and S.\ Klainerman (cf.\ \cite{CHRKLA:1993}), %
S.\ Brendle investigated in \cite{Brendle:2002}
the stability of a
hyperplane satisfying the membrane equation %
in Minkowski space.
In contrast to the work of H.\ Lindblad (cf.\ \cite{Lind:2004}), 
S.\ Brendle requires less regularity
of the initial data, comparable to that in \cite{CHRKLA:1993}.

Let us state the main problem addressed by this paper in non-technical
terms. \\
Let $N^{n+1}$ be an $(n+1)$-dimensional Lorentzian manifold,  
  and let $\Sigma_0$ be an $m$-dimensional spacelike
  submanifold of $N$. %
  Assume $\nu$ to be an unit timelike future-directed vector field normal to
  $\Sigma_0$. %
  \begin{center}
  \parbox{0.93\textwidth}{
   \textbf{Existence} \\
  Find an open $(m + 1)$-dimensional %
  timelike submanifold $\Sigma$ of $N$
  satisfying
  \begin{equation}
    \label{eq:geom_problem}
    H(\Sigma) \equiv 0,~ \Sigma_0 \subset \Sigma,
    \mbox{ and }\nu \mbox{ is tangential to }\Sigma.
  \end{equation}
  \textbf{Uniqueness} \\
  Show that for  $\Sigma_1$ and $\Sigma_2$ solving the IVP 
  \eref{eq:geom_problem}, 
  there exists an open set $V$ with
  $\Sigma_0 \subset V \subset N$ such that 
  \begin{equation*}
    V \cap \Sigma_1 = V \cap \Sigma_2.
  \end{equation*} 
}
\end{center}  
We will refer to $N, ~ \Sigma_0$, and $\nu$ as ambient manifold, initial
submanifold, and initial direction, respectively.
The main result of this work is an affirmative answer to the problem above
under suitable conditions. Here, the ambient manifold is assumed to admit
a time function in a neighborhood of the initial submanifold.
No further condition on the causal structure is required.

After fixing some notation in section \ref{sec:notation}, we perform the 
reduction of the  membrane
equation to a system of quasilinear
hyperbolic equations by employing  a gauge used in the context of the
Einstein equations (cf. \cite{FriRe:2000}).
In section \ref{indpar} we consider solutions of the membrane equation given
by parametrized immersions. In section \ref{main} these local
solutions are glued together to obtain a solution to the geometric Cauchy 
problem \eref{eq:geom_problem}.
The main existence result, Theorem \ref{thm:main_ex}, shows that
for non-compact regularly immersed submanifolds of any codimension in
an ambient manifold of any dimension, existence and uniqueness
holds. Uniform assumptions lead to a lower bound on the length of timelike 
curves lying in the world volume, which gives rise to a geometric notion of
time of existence of a solution to the Cauchy problem \eref{eq:geom_problem}.
In Corollary \ref{cor:main_ex} it is shown that smooth  %
data lead
to smooth solutions of the membrane equation.
In Theorem \ref{thm:main_uni} we present the main uniqueness result for the 
membrane equation in the case of smooth data.

The results are taken from the author's Ph.D. thesis; for further 
reference see
\cite{Mil:2008}.

\label{sec:prel}
\section{Notation}
\label{sec:notation}
Let $N^{n+1}$ be an $(n+1)$-dimensional  smooth manifold endowed with
a Lorentzian metric $h$. 
The Levi-Civita connection with respect to $h$ and the corresponding Christoffel
symbols are denoted by $\D$ and
$\G$, respectively. 
Covariant derivatives of order $\ell$ 
are denoted by $\D^{\ell}$.
The $(1,3)$ or the $(0,4)$-version 
of the curvature of $h$ is denoted by $\RRiem$.

Suppose $\Sigma^{m+1}$ is an $(m+1)$-dimensional timelike submanifold
of $N$.
The metric and connection induced on $\Sigma$ are then given by
\begin{equation*}
  g := h\bigr|_{\Sigma} \quad\mbox{and}\quad
  \nabla_X Y := (\D_X Y)^{\top} \mbox{ for vector fields } X, Y \mbox{ tangent
    to }
  \Sigma.
\end{equation*}
The Christoffel symbols of $\nabla$ are denoted by $\Gamma$.
The second fundamental form and the mean curvature of $\Sigma$ are given by
\begin{equation}
  \label{eq:2nd_mean}
  \IIfull(X,Y) := (\D_X Y)^{\bot}\qquad\mbox{and}\qquad H = %
      \tr_g \IIfull.
\end{equation}
The usual definition of the mean curvature involves a factor 
dependent on the 
dimension, which was omitted here since we are only interested in the 
homogeneous equation.

The following definitions introduce types of initial submanifolds.
We refer to them as the \emph{immersion type} of the initial submanifold.
\begin{defn}
  \label{defn:immersed_smf}
  Let $M^m$ be an $m$-dimensional manifold, and let $\varphi: M \rightarrow N$
  be an immersion. The image of $\varphi$, $\Sigma_0:= \im \varphi$,
  is called a \emph{regularly immersed submanifold of $N$}.
\end{defn}
\begin{defn}
  \label{defn:loc_emb}
  Let $\Sigma_0$ be a regularly immersed submanifold of dimension $m$
  with immersion
  $\varphi: M^m \rightarrow N$.
  The set
  $\Sigma_0$ is called a \emph{locally embedded submanifold} if
  for every point $q \in \Sigma_0$ there exist open sets
  $q \in V \subset N$ and $U \subset M$ such that
  \begin{equation}
    \label{eq:defn_loc_emb}
    \varphi: U \rightarrow \varphi(U) \mbox{ is a diffeomorphism and} \quad
    \varphi^{-1}(V \cap \Sigma_0) = U.
  \end{equation}
\end{defn}
\begin{defn}
  \label{defn:fin_loc_emb}
  Let $\Sigma_0$ be a regularly immersed submanifold of dimension $m$
  with immersion
  $\varphi: M^m \rightarrow N$.
  The set
  $\Sigma_0$ is called a \emph{regularly immersed submanifold with finite 
    intersections} if
  for every point $q \in \Sigma_0$ there exist a neighborhood
  $V \subset N$ of $q$ and finitely many open disjoint sets $U_{\ell} \subset M$ 
  such that
  \begin{equation}
    \label{eq:defn_fin_loc_emb}
    \fl \varphi: U_{\ell} \rightarrow \varphi(U_{\ell}) 
    \mbox{ is a diffeomorphism for every $\ell$ and} \quad
    \varphi^{-1}(V \cap \Sigma_0) = \bigcup U_{\ell}.
  \end{equation}
\end{defn}
Let $\Sigma_0$ be a regularly immersed submanifold with immersion
$\varphi:M^m \rightarrow \Sigma_0 \subset N$. We make use of the 
induced metric $\ig = \varphi^{\ast} h$ on $M$.
The Levi-Civita-connection compatible with
$\ig$ %
is 
denoted by
$\inab$.
The connection on the pull-back bundle $\varphi^{\ast} TN$ is denoted
by $\hn$, and members of $\varphi^{\ast} TN$ are called 
\emph{vector fields along $\varphi$}.
Using this notation we introduce the \emph{second fundamental form of $\varphi$}
defined as the vector field
\begin{equation}
  \label{eq:init_2nd}
  \II(X,Y) = (\hn_X d\varphi(Y))^{\bot}_{\im \varphi},
\end{equation}
where $d\varphi$ denotes the differential of $\varphi$, and
$(\,.\,)^{\bot}_{\im \varphi}$ denotes the part normal to the image of 
$\varphi$.
If $\Sigma_0$ is a locally embedded submanifold
then this definition is independent of the local embedding and gives
the usual definition of second fundamental form of a submanifold.

Let $E$ be a Riemannian metric defined on $N$. For a vector field $\xi$ along 
$\varphi$ we introduce the following
norm
\begin{equation}
  \label{eq:gE_norm}
  \abs{\hn^{\ell} \xi}_{\ig,E}^2 = 
  \ig^{i_1 j_1} \cdots \ig^{i_{\ell} j_{\ell} }E_{AB} \hn_{i_1, \dots, i_{\ell} }\xi^A
  \hn_{j_1, \dots, j_{\ell}} \xi^B
\end{equation}
with the abbreviation $\hn_{i_1, \dots, i_{\ell} } = \hn_{i_1} \cdots \hn_{i_{\ell}}$.
Here, we used the following set of indices for coordinates.
On $M$ and on the initial submanifold $\Sigma_0$, coordinates
carry small Latin indices running from $1$ to $m$.
Capital Latin indices as $A,B,C,\dots$ 
will run from $0$ to $n$, and indicate coordinates on $N$. 
Our convention for the signature of a Lorentzian metric is 
$(\,{-}\,{+}\,\cdots\,{+}\,)$.
Partial derivatives in coordinates are abbreviated by 
$\partial_A = \partial/\partial y^A$. %
Local coordinates on the timelike submanifold $\Sigma$ carry
Greek indices running from $0$ to $m$.

An assumption on the causal structure of the ambient manifold 
will be made by using the following terminology.
\begin{defn}
  A real-valued function $\tau$ on $N$ is called \emph{time function}
  if its gradient $\D \tau$ is everywhere timelike. 
  Such a time function induces a \emph{time foliation} of $N$ by 
  its levelsets $\slice_{\tau}$, which are spacelike 
  hypersurfaces. %

  We further introduce the \emph{lapse $\psi$ of the foliation}
given by $\tau$
by 
\begin{equation}
  \label{eq:def_lapse}
  \psi^{-2} = - h(\D \tau, \D \tau),
\end{equation}
and the future-directed unit normal to the time foliation
is given by
\begin{equation*}
  \hT = - \psi \D\tau.
\end{equation*}
\end{defn}
The dual vector field $\partial_{\tau}$ to 
the differential $d\tau$ of the time function is given by
\begin{equation*}
  \partial_{\tau} = - \psi^2 \D\tau \quad\Rightarrow\quad \partial_{\tau}
  = \psi \hT.
\end{equation*}
The existence of such a function provides us with a possibility
of introducing a Riemannian
metric.
\begin{defn}
  \label{defn:def_E}
  Suppose $N$ admits a time function $\tau$.
  We consider the 
  decomposition into tangential and normal part with respect to the slices,
  denoted by $(\,.\,)^{\top}$ and $(\,.\,)^{\bot}$, respectively.
  We introduce a Riemannian metric $E$ associated to the time foliation
  by a change of the sign of the unit normal $\hT$ as follows %
  \begin{equation}
    \label{eq:def_E}
    E(v,w) := h\bigl(v^{\top}, w^{\top}\bigr) + \lambda \mu, \quad\mbox{where }
    v^{\bot} = \lambda \hT\mbox{ and } w^{\bot} = \mu \hT.
  \end{equation}
  Since the slices are spacelike this metric is Riemannian.
\end{defn}
We use the norm on tensors induced by the metric $E$ defined as follows
\begin{equation*}
  \abs{T}^2_E = E^{A_1 C_1} \cdots E^{A_k C_k} \, E_{B_1 D_1} \cdots E_{B_{\ell} D_{\ell}}\,
    T_{A_1, \dots, A_k}^{B_1, \dots, B_{\ell}} \, T_{C_1, \dots, C_k}^{D_1, \dots, D_{\ell}}.
\end{equation*}

\section{Reduction}
\label{sec:reduction}
In this section we present a method to obtain a strictly hyperbolic
system from the membrane equation by using parametrizations and a specific
gauge condition.

Let $F:W \subset \rr \times M \rightarrow (N,h)$ denote an immersion
with induced metric $g:= F^{\ast} h$. We denote the Christoffel symbols
with respect to $g$ by $\Gamma$. The image of $F$ will correspond to 
a solution $\Sigma$ of the geometric IVP \eref{eq:geom_problem}. 

Let $x^{\mu}$
be coordinates on $ \rr\times M$, and let $y^A$ be coordinates on $N$.
We denote the covariant derivative on
the pullback bundle $F^{\ast}TN$ by $\hn^F$.
The membrane equation for the image of $F$ then reads in this situation  
  \begin{eqnarray}
    \label{eq:mem_geom}
    &&g^{\mu\nu}\hn^F_{\mu} \partial_{\nu} F - g^{\mu\nu} \bigl(\hn^F_{\mu}
  \partial_{\nu} F \bigr)_{\im F}^{\top}  = 0 %
  \\ 
  \label{eq:membrane}
  \fl\mbox{or, equivalently} \quad&& %
  g^{\mu\nu} \partial_{\mu} \partial_{\nu} F^A  + g^{\mu\nu} \partial_{\mu} F^B
  \partial_{\nu} F^C \G_{BC}^A(F) - 
  \Gamma^{\lambda} \partial_{\lambda} F^A  = 0,
\end{eqnarray}
where $(\,.\,)^{\top}_{\im F}$ denotes the part tangential to the image of $F$.

The hyperbolicity of \eref{eq:membrane} comes from the 
signature of the metric induced by $F$ on $W$. 
The degeneracy of \eref{eq:membrane} is a consequence of the
invariance %
under tangential diffeomorphisms of a solution.
It becomes manifest in the contracted Chrisoffel symbols induced on $W$,
which contain a term involving second-order derivatives of 
$F$ cancelling
the tangential part of the
leading second-order term.

To overcome the problem of degeneracy %
it is necessary
to fix a gauge.
We consider the so-called harmonic map gauge taken from \cite{FriRe:2000},
where it is used to show well-posedness of %
the Einstein equations.

Recall that a map $v: (M_1, g_1) \rightarrow (M_2, g_2)$ between two 
pseudo-Riemannian manifolds is called
\emph{harmonic map} (or \emph{wave map}) if it satisfies 
\begin{equation}
  \label{eq:harm_map}
  g_1^{ij} \partial_i \partial_j v^a - g_1^{ij} \Gamma(g_1)_{ij}^k \partial_k v^a
  + g_1^{ij} \partial_i v^b
  \partial_j v^c \Gamma(g_2)_{bc}^a = 0.
\end{equation}
\begin{defn}
  A solution $F: W \subset \rr \times M \rightarrow N$ of the membrane 
  equation \eref{eq:membrane}
  is in \emph{harmonic map gauge},
  if the following condition is satisfied
  \begin{equation}
    \label{eq:harm_cond}
    \id: \bigl(W, g = F^{\ast} h\bigr) 
    \rightarrow \bigl(W, \hat{g}\bigr)
    \quad\mbox{is a harmonic map,}
  \end{equation}
  where $\hat{g}$ is a fixed background metric. %
  We call this condition  \emph{harmonic map gauge condition}.
\end{defn}
\begin{rem}
  In coordinates, condition \eref{eq:harm_cond} reads 
\begin{equation}
  \label{eq:harm_cond_coord}
  g^{\mu\nu}(\Gamma_{\mu\nu}^{\lambda} - 
\hat{\Gamma}_{\mu\nu}^{\lambda}) = 0.
\end{equation}
Note that this is the trace of the difference of two connections,
and is therefore independent of coordinates.
\end{rem}
By inserting  the harmonic map gauge condition \eref{eq:harm_cond_coord}
into \eref{eq:membrane} we obtain the \emph{reduced membrane
equation}
\begin{equation}
  \label{eq:mem_red}
  g^{\mu\nu} \partial_{\mu} \partial_{\nu} F^A  + g^{\mu\nu} \partial_{\mu} F^B
  \partial_{\nu} F^C \G_{BC}^A(F) - 
  g^{\mu\nu}\hat{\Gamma}_{\mu\nu}^{\lambda} 
  \partial_{\lambda} F^A = 0.
\end{equation}
The reduced equation is strictly hyperbolic since %
the contracted Christoffel symbols  $g^{\mu\nu} \Gamma_{\mu\nu}^{\lambda}$
are replaced by a lower-order term.
We now show equivalence of the equations \eref{eq:mem_red} and 
\eref{eq:membrane}.
\begin{lem}
  \label{lem:equiv}
  The membrane equation \eref{eq:membrane} together with condition 
  \eref{eq:harm_cond_coord} is equivalent to \eref{eq:mem_red}.
\end{lem}
\begin{proof}
  We only need to show that the reduced membrane equation 
  \eref{eq:mem_red} implies
  \eref{eq:membrane} and condition \eref{eq:harm_cond_coord};
  the other direction follows immediately.
  Suppose the reduced membrane equation
  \eref{eq:mem_red} holds. 
  We compute the contracted Christoffel symbols with respect to\
  the  metric induced by a solution
  \begin{equation*}
    g^{\mu\nu} \Gamma_{\mu\nu}^{\lambda} 
     =  g^{\kappa\lambda} g^{\mu\nu} \hat{\Gamma}_{\mu\nu}^{\delta} \partial_{\delta}
    F^A h_{AD} \partial_{\kappa} F^D = g^{\mu\nu} \hat{\Gamma}_{\mu\nu}^{\lambda}.
  \end{equation*}
  Using this identity together with %
  \eref{eq:mem_red} gives us the desired result.
\end{proof}

\begin{rem}
  This result can be derived in a somewhat more abstract fashion
  if we use \eref{eq:mem_geom} instead of \eref{eq:membrane}.
  The reduced membrane equation \eref{eq:mem_red} can be 
  expressed
  \begin{equation*}
    g^{\mu\nu}\hn^F_{\mu} \partial_{\nu} F = g^{\mu\nu} 
  \hat{\Gamma}_{\mu\nu}^{\lambda} \partial_{\lambda} F.
  \end{equation*}
  The right hand side of the latter equation is tangential to the image
  of $F$, therefore the
  normal part of the left hand side %
  has to vanish, and the tangential part has to
  coincide with the right hand side which yields the desired identity
  \eref{eq:harm_cond_coord}.
\end{rem}

\begin{rem}
  Other gauges corresponding to the prescription of the contracted
  Christoffel symbols are possible. Therefore, we can mine the rich
  repertoire of gauges for the Einstein equations (cf.\ \cite{FriRe:2000}).
\end{rem}

\begin{rem}
  In contrast 
  to the case of the Einstein equations the reduction process needs
  no further equation. 
\end{rem}

Based on \eref{eq:membrane} we  consider the following
problem involving
existence and uniqueness. \\
Let $\Sigma_0$ be a regularly immersed $m$-dimensional 
  submanifold  of $N$
  with immersion
  $\varphi: M^m \rightarrow \Sigma_0 \subset N$. Let $\alpha$ be a function on
  $M$ called \emph{initial lapse}, and let
  $\beta$  be a vector field on $M$ called \emph{initial shift}. 
  Suppose $\nu: M \rightarrow TN$ is an unit timelike future-directed
  vector field along $\varphi$ normal to $\Sigma_0$.
  \begin{center}
  \parbox{0.93\textwidth}{
     \textbf{Existence} \\
  Find a neighborhood $W\subset \rr\times M $ 
  of $\{0\} \times M$ and
  an immersion $F:W \rightarrow N$
  solving the IVP 
  \begin{equation}
    \label{eq:param_ivp}
    H(\im F)\equiv 0,~\restr{F} = \varphi,~\restr{\partial_t F} 
    = \alpha \,\nu %
    + d\varphi(\beta)
  \end{equation}
  such that $\partial_t F$ is timelike, and the image of $F(t): M \rightarrow N$
  is spacelike. The last term consists of the differential 
  of $\varphi$ applied to the vector field $\beta$.
  The parameter $t$ denotes the first component of $\rr \times M$. \\
  \textbf{Uniqueness} \\
  Suppose $\Sigma_0$ is locally embedded, and let $\nu$ be an unit timelike
  future-directed vector field
  on $\Sigma_0$. Assume 
  $\bar{\varphi}: M \rightarrow N$ to be an immersion with 
  $\Sigma_0 = \im \bar{\varphi}$, and let $\bar{\alpha}$ and $\bar{\beta}$
  be another choice of initial lapse and shift.\\  
  Show that for immersions $F, \bar{F}: W\subset \rr\times M \rightarrow N$
  solving \eref{eq:membrane}, and attaining the initial values
  \begin{eqnarray*}
    \restr{F}  = \varphi,~\restr{\partial_t F} 
    = \alpha \,\nu \circ \varphi 
    + d\varphi(\beta), 
    \\
    \restr{\bar{F}} = \bar{\varphi},~\restr{\partial_t \bar{F}} 
    = \bar{\alpha} \,\nu \circ \bar{\varphi} 
    + d\bar{\varphi}(\bar{\beta}),
  \end{eqnarray*}
  there exists a local diffeomorphism $\Psi$ of $W$ %
  such that
  $F\circ \Psi^{-1} = \bar{F}$.   
}
\end{center} 

\begin{rem}
Throughout the following we use a \emph{special background metric} 
$\hat{g}$ defined on $\rr \times M$.
If the initial values of $F$ are %
$\restr{F} = \varphi$
and $\restr{\partial_t F} = \alpha \,\nu %
+ d\varphi(\beta)$, then
we define
\begin{equation}
  \label{eq:back_metric}
  \hat{g}  := - \alpha^2 dt^2 +
    \ig_{ij} (\beta^i dt + dx^i)(\beta^j dt + dx^j).
\end{equation}
Observe that $\ig_{ij}, \alpha$, and $\beta$ are independent of the 
parameter $t$.
\end{rem}

\section{Parametrized Immersions}
\label{indpar}
This section is devoted to discuss the Cauchy problem \eref{eq:param_ivp}.
Solutions will be obtained as parametrization through spacelike immersions.
First we present the main assumptions of this
paper. 
The main results of this section are
stated in Propositions \ref{prop:immersion_ex} and %
\ref{prop:param_uni}.

\subsection{Assumptions}
\label{assum}
We use the assumptions on the date presented in this section
if not explicitly stated otherwise. In order to get 
results which are independent of the scale of the ambient manifold 
--- multiplication
of the metric with the square of a positive constant ---
 we introduce a constant $R > 0$. The purpose of this constant is to absorb
scaling such that scale-invariant constants occur in the assumptions.

Let $s > m/2 + 1$ be an integer, and let $V\subset N$ and 
$U \subset M$
denote open sets. 
\\
\textbf{Assumptions on the ambient manifold:}
 We suppose $N$ to admit a time function $\tau$ in %
 $V$. Assume the metric $h$ and $\tau$ to be $C^{s+3}$.
 Let $\psi$ denote
  the lapse of the time foliation  induced by $\tau$ 
  defined
  by \eref{eq:def_lapse}, and let $E$ denote 
  the %
  Riemannian metric introduced in Definition \ref{defn:def_E}.

  We suppose there are constants $C_1, C_2,
  C_N$, and $C_{\tau}$ independent of $R$ such that the following 
  inequalities
  are uniformly satisfied in $V$
  \begin{eqnarray*}
    \fl C_1 \le  R^{-1}\psi \le C_2,~
       \sum_{\ell = 0}^{s+1}
       R^{2+\ell}\abs{\D^{\ell} \RRiem}_E \le C_N, 
    \quad\mbox{and}\quad 
    \sum_{\ell = 1}^{s+2} R^{1+\ell}\abs{\D^{\ell} (\D\tau)}_E 
    \le C_{\tau}, %
  \end{eqnarray*}
  where $\D(\D \tau)$ denotes the $(1,1)$-tensor
    obtained by applying the covariant derivative to the gradient of the time 
    function $\tau$. 
\\
\textbf{Assumptions on the initial submanifold:}
Assume $\varphi$ to be $C^{s+2}$. %
Let $N$ admit a time function in a neighborhood of $\varphi(U)$.
 We suppose there exist constants $\omega_1, C_{\varphi}$ 
  independent of $R$ such that the following inequalities
  are uniformly satisfied in $U$
    \begin{eqnarray*}
         \inf\bigl\{ -h( \gamma , \hT ) : \gamma \mbox{ unit timelike 
           future-directed 
           normal} \mbox{ to }\Sigma_0
      \bigr\}\le \omega_1,
      \\
      \sum_{\ell = 0}^{s} R^{\ell + 1}\abs{\hn^{\ell} \II}_{\ig,E}  \le C_{\varphi},
    \end{eqnarray*}
    where $\hT$ denotes the unit timelike future-directed normal to
    the time foliation on $N$, and 
    $\II$ denotes the second fundamental form of $\varphi$ %
    (cf.\ \eref{eq:init_2nd}).
\\
\textbf{Assumptions on the initial direction:}
Assume $\nu$ to be $C^{s+1}$.
 We suppose there exist constants $\omega_2, C_{\nu}$ independent of $R$ 
  such that the following inequalities
  are uniformly satisfied in $U$
    \begin{equation*}
      - h(\nu, \hT) \le \omega_2
      \qquad\mbox{and}\qquad
      \sum_{\ell = 1}^{s +1} R^{\ell}\babs{\hn^{\ell} \nu}_{\ig,E} 
      \le C_{\nu}. %
    \end{equation*}
\begin{rem}
  If the initial submanifold $\Sigma_0$ is locally embedded, and
  the initial direction is defined on $\Sigma_0$ rather than on 
  the domain of the embedding, 
  then 
  the assumptions on $\Sigma_0$ and $\nu$ do not depend on the 
  local embedding. %
\end{rem}
When considering the initial value problem \eref{eq:param_ivp} for immersions we
have the addinitional assumptions on initial lapse and shift.
\\
\textbf{Assumptions on initial lapse and shift:}
Assume $\alpha$ and $\beta$ to be $C^{s+1}$.
We suppose there exist constants $L, ~C_{\alpha}$, and  $C_{\beta}$ independent of
$R$ 
such that the following inequalities
  are uniformly satisfied in $U$
    \begin{eqnarray*}
      R^2(- \alpha^2 + \abs{\beta}_{\ig}^2) \le - L, 
      \\
        \sum_{\ell = 1}^{s+1} R^{\ell} \abs{\inab^{\ell} \beta}_{\ig}  \le C_{\beta},
      \quad
      \mbox{and}\quad \sum_{\ell = 0}^{s+1} R^{\ell} \abs{\inab^{\ell} \alpha}_{\ig}  
      \le C_{\alpha}. 
    \end{eqnarray*}   

\subsection{Existence}
\label{sec:diff_lapse}
The following result gives a solution to the IVP \eref{eq:param_ivp}.
For a detailed proof we refer to \cite{Mil:2008}.

\begin{prop}
  \label{prop:immersion_ex}
  Suppose for each $p\in M$ there exist a neighborhood $U \subset M$ of 
  $p$, and
  an open set $V \subset N$ %
  containing $\varphi(U)$ 
  such that the assumptions described in the previous section
  are satisfied in $V$ and $U$.

  Then there exist a neighborhood $W \subset \rr\times M$
  of $\{0\} \times M$ and a $C^2$-immersion $F:W\rightarrow N $
  such that $\partial_t F$ is timelike, and the image of $F(t):M \rightarrow N$
  is spacelike
  solving the reduced membrane equation \eref{eq:mem_red}
  with respect to\  the background metric $\hat{g}$ defined by
  \eref{eq:back_metric},
  and attaining the initial values
    $\restr{F} = \varphi$ and
    $\restr{\partial_t F} = \alpha \,\nu %
    + d\varphi(\beta)$.
\end{prop}

\begin{rem}
  \label{rem:ex_loc}
  The proof involves solving the reduced membrane equation 
  \eref{eq:mem_red} in a neighborhood of each point $ p \in M$. Standard 
  local uniqueness
  results for hyperbolic equations then show that two local solutions coincide
  on their common domain. Therefore, a spatially local version of 
  Proposition 
  \ref{prop:immersion_ex} is also valid.
\end{rem}

\begin{rem}
\label{rem:smooth}
  Let $\ell_0$ be an integer, and suppose $s > m/2 + 1 + \ell_0$.
  Then the solution $F$ is of class $C^{2 + \ell_0}$.
\end{rem}

\begin{rem}
  \label{rem:W_T}
  If the assumptions are uniform in $p$, then the domain $W$ of the solution
  has the form $[-RC_0,RC_0] \times M$ with a constant $C_0 > 0$ independent
  of the scale.
\end{rem}

\subsection{Uniqueness}
Let $\Sigma_0$ be a locally embedded submanifold of $N$.
Suppose $\varphi, \bar{\varphi} : M \rightarrow N$ are two immersions
satisfying the conditions of Definition \ref{defn:loc_emb} for $\Sigma_0$.
The metrics induced on $M$ by $\varphi, \bar{\varphi}$ we denote
by $\ig, \mathring{\bar{g}}$, respectively.
Let $\psi_0: M \rightarrow M$ denote the local diffeomorphism defined by
  $\varphi\circ \psi_0^{-1} = \bar{\varphi}$.
The differential of $\psi_0$ is a member of $T^{\ast} M \otimes TM$,
and we use the following norm
\begin{equation*}
  \abs{d\psi_0}_{\ig,\mathring{\bar{g}}}^2 = 
  \ig^{ij} \mathring{\bar{g}}_{k\ell} (d\psi_0)_i^k (d\psi_0)_j^{\ell}. 
\end{equation*}
Norms of higher-order derivatives of $\psi_0$ are defined analogously.
Let $\alpha, \bar{\alpha}> 0$ be two functions on $M$, and
let $\beta, \bar{\beta}$ be two vector fields on $M$.
Assume $\nu$ to be %
defined on $\Sigma_0$.

We obtain the following uniqueness result.
\begin{prop}
  \label{prop:param_uni}
  Let $\varphi, \alpha, \beta$ and $ \bar{\varphi},
  \bar{\alpha}, 
      \bar{\beta}$ satisfy the assumptions of sections \ref{assum}.
  Suppose there
  exist constants $C_1^{\psi}, C_2^{\psi}$ independent of $R$ such that
  \begin{equation}
    \label{eq:cond_initial_diffeo}
   \abs{d\psi_0}_{\ig,\mathring{\bar{g}}} \le C_1^{\psi}\quad
    \mbox{and}\quad R\abs{d^2 \psi_0}_{\ig,\mathring{\bar{g}}} \le C_2^{\psi}.
  \end{equation}
  Let $F, \bar{F}: W \subset \rr \times M \rightarrow N$ be two %
  $C^{s+2}$-solutions of the membrane equation \eref{eq:membrane} 
  attaining the initial values
  \begin{eqnarray*}
    \restr{F}  = \varphi,~\restr{\partial_t F} 
    = \alpha \,\nu \circ \varphi 
    + d\varphi(\beta), 
    \\
    \restr{\bar{F}} = \bar{\varphi},~\restr{\partial_t \bar{F}} 
    = \bar{\alpha} \,\nu \circ \bar{\varphi} 
    + d\bar{\varphi}(\bar{\beta}).
  \end{eqnarray*}
  Let $\hat{\bar{g}}$ be
  the background  metric defined by \eref{eq:back_metric} using the 
  initial values of $\bar{F}$.
  Assume $\bar{F}$ to be in 
  harmonic
  map gauge with respect to\   $\hat{\bar{g}}$.
  
  Then there exists a local diffeomorphism $\Psi$ of $M$
  such that $F \circ \Psi^{-1}$ and $\bar{F}$ coincide on a neighborhood
  of $\{0\} \times M$.
\end{prop}

From \cite{Mil:2008} we have the following uniqueness result
for solutions to the IVP \eref{eq:param_ivp}, which are in harmonic map
gauge with respect to\ the background metric defined by \eref{eq:back_metric}.
\begin{prop}
  \label{prop:mem_red_uni}
  Let the assumptions of section \ref{assum} be satisfied.
  Assume $F:W \rightarrow N$ and $\bar{F}: \widetilde{W} \rightarrow N$
  to be two immersions of class $C^2$ solving
  the IVP \eref{eq:param_ivp}. Suppose both solutions  are in harmonic map 
  gauge with respect to\
  the background metric defined by \eref{eq:back_metric}.

  Then there exists a neighborhood of $\{0\} \times M$, on which the
  two solutions $F$ and $\bar{F}$ coincide.
\end{prop}

\begin{rem}
  If the domains $W$ and $\widetilde{W}$ have the form
  $[-RC_0,RC_0] \times M$, then there exists a scale-invariant
  constant $C$ such that $F$ and $\bar{F}$ coincide
  for $-RC \le t \le RC$.
\end{rem}

\begin{rem}
  \label{rem:uni_loc}
  As is the case for the existence result stated in Proposition 
  \ref{prop:immersion_ex},
  a spatially local version of Proposition \ref{prop:mem_red_uni}
  holds as well. From uniqueness results for hyperbolic equations, we 
  obtain
  that uniqueness holds on cones.
\end{rem}
The previous result relies on the fact that both solutions are in harmonic
map gauge with respect to\  the same background metric.
The following proposition gives a condition under which it is possible
to align the harmonic map gauge of two solutions. 
\begin{prop}
  \label{prop:cond_diffeo}
  Suppose in the situation of Proposition \ref{prop:param_uni}
  there exist open sets 
  $W,\widetilde{W} \subset \rr\times M$ such that for a fixed point 
  $p \in M$ it holds
  $(0,p) \in W$ and $\bigl(0,\psi_0(p)\bigr) \in \widetilde{W}$. 
  Suppose there exists a diffeomorphism
  $\Psi  : W \rightarrow \widetilde{W}$
  such that
     \begin{equation}
       \label{eq:diffeo_harm}
     \Psi: \bigl(W, g= F^{\ast}h\bigr)
     \rightarrow \bigl(\widetilde{W}, \hat{g}
     \bigr) 
     \mbox{ is a harmonic map}.
   \end{equation}
      Furthermore, assume the inverse to satisfy the initial conditions
   \begin{eqnarray}
     \label{eq:def_inverse_trafo_initial}
     &&\Psi^{-1}\bigr|_{\tilde{t} = 0} = (0,\psi_0^{-1}) \quad\mbox{and}\quad
     \partial_{\tilde{t}} \Psi^{-1}\bigr|_{\tilde{t} = 0}
     = \hat{\lambda} \partial_t + \hat{\chi}
     \\ 
     \fl\mbox{with }\qquad&&
     \hat{\lambda}(p) = \case{\bar{\alpha}(p)}{\alpha(\psi_0^{-1}(p))}
     \quad\mbox{and} \quad
     \hat{\chi}(p) = d(\psi_0^{-1})_p(\bar{\beta}) 
     - \hat{\lambda}(p) \beta[\psi_0^{-1}(p)]. \nonumber
   \end{eqnarray}
   Then $F\circ \Psi^{-1}$ and $\bar{F}$ coincide on a neighborhood of
   $(0,p)$.
 \end{prop}
 \begin{proof}
   From the harmonic map equation satisfied by $\Psi$ we derive the
   condition for the solution $F \circ \Psi^{-1}$ of the membrane
   equation to be in harmonic map gauge with respect to\  the background metric 
   $\hat{\bar{g}}$ (cf.\ \eref{eq:harm_cond}). 

    It remains to show that the initial values of $\Psi^{-1}$
   give us the appropriate initial values for $F \circ \Psi^{-1}$.
   From the definition of $\psi_0$ we derive that 
   $\restr{F \circ \Psi^{-1}} = \bar{\varphi}$.
   To obtain the initial value for the velocity we compute
   \begin{eqnarray*}
         \partial_{\bar{t}}(F   \circ \Psi^{-1})\bigr|_{\tilde{t} = 0}(p) 
     &=& \hat{\lambda}(p) 
     \alpha[\psi_0^{-1}(p)] \nu\circ \varphi[\psi_0^{-1}(p)]
     + \hat{\lambda}(p) 
     d\varphi_{\psi_0^{-1}(p)}(\beta) 
     \\
     &&{}+ d(\varphi \circ \psi_0^{-1})_p(\bar{\beta})  
     - \hat{\lambda}(p) d\varphi_{\psi_0^{-1}(p)}(\beta)
   \end{eqnarray*}
   Since the second and the last term cancel, the claim follows from
   the definition of the initial values of $\Psi^{-1}$.
   Therefore,  $F \circ \Psi^{-1}$ and $\bar{F}$ satisfy the reduced membrane
   equation \eref{eq:mem_red} with respect to\  the background metric
   $\hat{\bar{g}}$ attaining the initial values of $\bar{F}$.   
   From Proposition \ref{prop:mem_red_uni} we obtain that both immersions 
   coincide locally.
 \end{proof}

To obtain Proposition \ref{prop:param_uni} 
it remains to show existence of a local diffeomorphism satisfying
condition \eref{eq:diffeo_harm}. Since the metrics $g = F^{\ast} h$ and
$\hat{g}$ are Lorentzian, the harmonic map equation for $\Psi$
has a structure similar to the reduced membrane equation \eref{eq:mem_red}.
The main difference is that the harmonic map equation is semilinear.
Hence, the same techniques can be applied as to obtain  a solution,
and we have the following result.
\begin{prop}
  \label{prop:ex_uni_diffeo}
  In the situation of
  Proposition \ref{prop:param_uni}
  there exist a neighborhood $\widehat{W}\subset \rr \times M$ of
  $\{0\} \times M$ and a
  $C^2$-immersion 
  $\Psi: \widehat{W}\rightarrow \rr \times M$ satisfying
  \eref{eq:harm_map} with $(M_1, g_1) = (W, g ) $ and 
  $(M_2, g_2) = (\rr \times M, 
  \hat{\bar{g}})$.
Furthermore,
$\Psi$ is invertible in a neighborhood of $\{0\} \times M$, and
the inverse
  attains the initial values 
  \eref{eq:def_inverse_trafo_initial}.
\end{prop}

\section{Geometric results}
\label{main}
In this section we  consider existence and uniqueness for the membrane
equation in purely geometric terms.
An existence result is obtained including a geometric notion of 
time of existence of a solution. The uniqueness result shows that
two submanifolds solving the IVP coincide on a neighborhood
of the initial submanifold.
\subsection{Existence}
The next definition gives a notion of ``time of existence'' which respects the 
geometric
behavior of a solution; %
note this does not necessarily coincide with the time parameter
of Proposition \ref{prop:immersion_ex}.
\begin{defn}[Time of existence]
  \label{def:proper_time}
  Let $\Sigma$ be a solution of the IVP \eref{eq:geom_problem}.
  The \emph{time of existence} $\tau_{\Sigma}$ of $\Sigma$ is given by
  \begin{eqnarray}
      \tau_{\Sigma} :=
      \inf_{p \in \Sigma_0}\sup\{\mbox{\textup{length }} 
      \mbox{\textup{of all timelike }} & \mbox{\textup{future-directed}} \nonumber
        \\
        &
      \mbox{\textup{curves in $\Sigma$ emanating from }} p \}. \nonumber
  \end{eqnarray}
\end{defn}
Our aim is an existence result for the IVP \eref{eq:geom_problem} which
includes a lower bound on the time of existence. Uniform assumptions imposed on 
ambient manifold,
initial submanifold, and initial direction provide us with
such a lower bound. 
\begin{thm}
  \label{thm:main_ex}
 Let $\rho > 0$ be a constant. 
 Suppose for each $q \in \Sigma_0$ there exists a neighborhood $V \subset N$
 of $q$ such that
   $B^E_{R\rho}(q) \subset V$, and the assumptions of section \ref{assum}
   are satisfied in $V$ and $\varphi^{-1}(V \cap \Sigma_0)$ %
  with constants independent of $q$.

  Then there exists an open $(m + 1)$-dimensional regularly immersed
  Lorentzian submanifold
  $\Sigma$ of class $C^2$ satisfying the IVP \eref{eq:geom_problem}.
  Furthermore, there exists a scale-invariant constant $\delta>0$ such that
  \begin{equation*}
    \tau_{\Sigma} \ge R\,\delta. %
  \end{equation*}
\end{thm}

\begin{proof} %
  From Proposition \ref{prop:immersion_ex} and Remark \ref{rem:W_T} we obtain
  an immersion $F:[-T,T] \times M\rightarrow N$ of class $C^2$
  solving the IVP \eref{eq:param_ivp} with initial lapse equal 1
  and initial shift equal 0. 
  Letting $\Sigma:= \im F$
  gives us a regularly immersed timelike submanifold solving the
  IVP \eref{eq:geom_problem}.

  The construction of $F$ provides us with a lower bound on the time
  parameter $T$, and an estimate for the timelikeness of $\partial_t F$. Hence, a 
  lower bound
  exists for the length of the timelike curves $t \mapsto F(t,p)$ for
  all $p \in M$.
\end{proof}

\begin{rem}
  \label{rem:ex_main_non_uniform}
  The proof shows that the theorem applies to the situation where
  the assumptions are valid only locally.
\end{rem}
The following corollaries show that a solution to the IVP 
\eref{eq:geom_problem}
for the membrane equation can be constructed in such a way that 
the immersion type %
of the initial submanifold is preserved.
\begin{cor}
  \label{cor:main_ex_loc_emb}
  Let $\Sigma_0$ be locally 
  embedded; for each point $q \in \Sigma_0$ let $U_q\subset M$ and $V_q
  \subset N$ 
  denote
  sets satisfying the conditions \eref{eq:defn_loc_emb} of Definition 
  \ref{defn:loc_emb}.
  Suppose that for each point $q \in \Sigma_0$ 
  the assumptions of section \ref{assum}
   are satisfied in  $V_q$ and $U_q$. 

  Then there exists a locally embedded timelike $(m+1)$-dimensional
  submanifold $\Sigma$
  of class $C^2$ solving the IVP \eref{eq:geom_problem}.
\end{cor}

\begin{proof}
  A local solution $F_q$ to the IVP \eref{eq:param_ivp} with initial
  lapse equal 1 and initial shift equal 0
   can be obtained by a local version of Proposition
  \ref{prop:immersion_ex} (cf.\ Remark \ref{rem:ex_loc}).
  This solution is defined on a neighborhood of 
  $\bigl(0,\varphi^{-1}(q)\bigr)$
  with values in $V_q$. From the inverse function theorem we obtain
  that the solution is an embedding in a neighborhood of 
  $\bigl(0,\varphi^{-1}(q)\bigr)$. By shrinking the domain $W_q$
  of $F_q$  further we achieve 
  that the set $W_q\cap ( \{0\} \times M)$
  lies in $U_q$. %
  Consider the family $(W_q)_q$ of the domains of all local solutions. 
  The intersection
  with $\{0\} \times M$ provides a covering of $M$.
  Choose
  points $q$ such that the intersection of the family of domains of the local 
  solutions
  with $\{0\} \times M$ form a locally finite covering subordinate to the
  above covering.
  Define a mapping which coincides with the local solutions on each
  set belonging to the locally finite subcovering above. 
  From a local version of Proposition \ref{rem:uni_loc} we derive that local 
  solutions
  coincide on their common domain.
  Therefore, this mapping is a well-defined immersion, 
  whose image is a locally embedded timelike submanifold of class $C^2$.
\end{proof}
To obtain a solution for regularly immersed initial submanifolds with
locally
finite intersections it is necessary to change the way in which the local 
solutions are pieced together.
\begin{cor}
  \label{cor:main_ex_finite}
   Let $\Sigma_0$ be regularly immersed with locally finite
  intersections; for each point $q \in \Sigma_0$ let $U_{q,\ell}\subset M$ and 
  $V_q \subset N$ denote
  sets satisfying the conditions \eref{eq:defn_fin_loc_emb} of Definition 
  \ref{defn:fin_loc_emb}.
  Suppose that for each point $q \in \Sigma_0$
 the assumptions of section \ref{assum}
   are satisfied in  $V_q$ and $U_{q,\ell}$ for every $\ell$.

  Then there exists a timelike $(m+1)$-dimensional
  regularly immersed submanifold $\Sigma$ with locally finite intersections 
  of class $C^2$ solving the IVP \eref{eq:geom_problem}.
\end{cor}

\begin{proof}
  For each $q \in \Sigma_0$ we pick the finitely many $U_{q,\ell}$
  satisfying \eref{eq:defn_fin_loc_emb}, and
  solve the membrane equation in $U_{q,\ell}$ with values in $V_q$.
  We shrink the domain of the solutions $F_{q,\ell}$ to obtain embeddings.
  By shrinking the 
  set $V_q $ to a subset $\tilde{V}_q \subset N$ we achieve
  that 
  \begin{equation*}
    \bigl(F_{q,\ell}\bigr)^{-1}\bigl(\tilde{V}_q \cap \im F_{q,\ell}\bigr)
    = W_{q,\ell} \subset \dom(F_{q,\ell}) \subset \rr \times U_{q,\ell}
    \quad\mbox{for all } \ell.
  \end{equation*}
  Consider the family $\mathcal{U} = 
  \bigl(\tilde{V}_q \cap \Sigma_0 \bigr)_{q \in \Sigma_0}$.
  Choose a locally finite covering $\bigl(\tilde{V}_{q_{\lambda}} \cap 
  \Sigma_0 \bigr)_{\lambda \in \Lambda}$
  of $\Sigma_0$
  subordinate to $\mathcal{U}$.
  Let $\tilde{U}_{q,\ell}$ denote the part of $W_{q,\ell}$ which belongs
  to $\{0\} \times M$.
  Consider the family $\bigl(\tilde{U}_{q,\ell} \bigr)_{q\in \Sigma_0}$.
  Then the family $(\tilde{U}_{q_{\lambda}, \ell})_{\lambda}$ is a locally finite
  covering of $M$ subordinate to $\bigl(U_{q,\ell} \bigr)_{q \in \Sigma_0}$
  due to the finiteness of the sets $U_{q, \ell}$ for fixed
  $q$.
  
  Let $F$ be a mapping defined by $F_{q,\ell}$ on $\tilde{U}_{q,\ell}$.
  Since the local solutions coincide on common domains, this mapping is 
  well-defined.   
  By construction it follows that $\Sigma := \im F$
  is  regularly immersed with locally finite intersections 
  and is of class $C^2$.
\end{proof}
We show that smooth data lead to a smooth solution of the IVP
\eref{eq:geom_problem} for the membrane equation again respecting the 
immersion type of 
the initial submanifold.
\begin{cor}
  \label{cor:main_ex}
  Assume $(N,h)$ to be  smooth, and suppose
  $\Sigma_0$ is either smoothly %
  \begin{enumerate}
  \item regularly immersed, 
  \item locally embedded, or
  \item regularly immersed  with locally finite intersections.
  \end{enumerate}
  Suppose $N$ admits a smooth time function $\tau$ in a neighborhood 
  of the initial submanifold $\Sigma_0$. Assume  $\nu$ to be smooth.
  
  Then there exists an open smooth $(m + 1)$-dimensional timelike
  submanifold $\Sigma$ which is
  \begin{enumerate}
  \item regularly immersed, 
  \item locally embedded, or 
  \item regularly immersed  with locally finite 
    intersections, 
  \end{enumerate}
  satisfying the IVP \eref{eq:geom_problem}.
\end{cor}

\begin{proof}
  For each integer $\ell_0$ it follows from the 
  smoothness
  of $h, \tau$, the immersion $\varphi$, and the initial direction that
  the assumptions of section \ref{assum} 
  are satisfied for $s > m/2 + 1 + \ell_0$ in
  a neighborhood of each point $q \in \Sigma_0$.
  Depending on the immersion type of $\Sigma_0$ 
  we apply Theorem 
  \ref{thm:main_ex}
  and Remark \ref{rem:ex_main_non_uniform},
  or Corollaries
  \ref{cor:main_ex_loc_emb} and \ref{cor:main_ex_finite}.
  We
  obtain a
  solution $\Sigma$ of class $C^{2 + \ell_0}$ by taking Remark 
  \ref{rem:smooth} into account.
\end{proof}

\subsection{Uniqueness}
In this section we consider the geometric uniqueness problem.
In Proposition \ref{prop:param_uni} it was shown that the construction of
a solution made in Proposition \ref{prop:immersion_ex} is independent of the 
choice of 
immersion of the
initial submanifold, initial lapse, and initial shift. Therefore, it remains to
construct an immersion of an arbitrary solution to \eref{eq:geom_problem}
which is in harmonic map gauge
with respect to\ the background metric defined by the initial values.%
\begin{thm}
  \label{thm:main_uni}
    Assume $(N,h)$ to be smooth, and suppose
  $\Sigma_0$ is smooth locally embedded.
  Let $N$ admit a smooth time function $\tau$ in a neighborhood 
  of the initial submanifold $\Sigma_0$. Suppose the initial direction
  $\nu$ is  smooth.

  Let $\Sigma_1$ and $\Sigma_2$ be two open smooth $(m + 1)$-dimensional
  locally embedded
  Lorentzian submanifolds of $N$ solving the IVP \eref{eq:geom_problem}.
  
  Then there exists a neighborhood $\Sigma_0 \subset V \subset N$ of $\Sigma_0$
  such that
  \begin{equation*}
    V \cap \Sigma_1 = V \cap \Sigma_2.
  \end{equation*}
\end{thm}

\begin{rem}
  In view of the fact that the image of a solution to \eref{eq:param_ivp}
  is a solution to \eref{eq:geom_problem} (cf.\ proof of Theorem 
  \ref{thm:main_ex}), Theorem \ref{thm:main_uni} and Proposition 
  \ref{prop:param_uni} show that
  initial lapse and shift can be given freely.
\end{rem}
Our strategy to prove this theorem is to compare an arbitrary solution
with the solution constructed in the previous section.
To apply the uniqueness result of Proposition  \ref{prop:param_uni} we need 
to construct
an immersion satisfying the IVP \eref{eq:param_ivp}.
\begin{prop}
  \label{prop:ex_embed}
  Let $(N,h),~\Sigma_0$, and $\nu$ satisfy the assumptions of Theorem
  \ref{thm:main_uni}.
  Let $\Sigma$ be a smooth locally embedded 
  solution
  to the IVP \eref{eq:geom_problem}.

  Then there exists an immersion  $F:W \subset
  \rr \times M \rightarrow N$ with $\Sigma \supset \im F  $  a locally 
  embedded submanifold.
  Furthermore, $F$ has the properties 
  that $\partial_t F$ is timelike, and $F(t): M \rightarrow N$ has a spacelike
  image. The initial values of $F$ are given by
  $\restr{F} = \varphi$ and $\restr{\partial_t F}
  = \nu \circ \varphi$. 
\end{prop}

\begin{proof}
  Let $p \in M$, and let $\gamma_{\varphi(p)}(t)$ be a geodesic in $\Sigma$
  attaining
  the initial values $\gamma_{\varphi(p)}(0) = \varphi(p)$ and
  $\dot{\gamma}_{\varphi(p)}(0) = \nu \circ \varphi(p)$.
  Set $F(t,p) = \gamma_{\varphi(p)}(t)$. Then
  $F$ is an immersion, since $\varphi$ is assumed to be an immersion,
  and $\nu$ is assumed to be unit timelike.
  The initial values of $F$ follow from the initial values of the geodesics.
  From the smoothness of the initial data it follows that $F$ is also
  smooth.

  The construction above is similar to Gaussian coordinates. 
  In an analogous way %
  it follows that $\partial_t F$ is timelike, and $F(t): M \rightarrow N$ has a 
  spacelike image. From the same argument we derive that
  the geodesics do not cross in a neighborhood of any point
  in $\Sigma_0$ as long as $\varphi$ is an embedding, and
  $\Sigma$ is an embedded submanifold around that point.
\end{proof}
We are now in the position to give a proof of the main uniqueness result.
\\
\textbf{Proof of Theorem \ref{thm:main_uni}:}\\
  Let $F_0$ be the solution to  the IVP \eref{eq:param_ivp}
  with initial values $\restr{F_0} = \varphi$ and $\restr{\partial_t F_0}
  = \nu \circ \varphi$ constructed in Corollary \ref{cor:main_ex}.
  Let $\widehat{\Sigma} := \im F_0$ denote the locally embedded image
  of the solution $F_0$.
  We compare a smooth solution $\Sigma$ to the
  IVP \eref{eq:geom_problem}
  with the solution $\widehat{\Sigma}$.

  From Proposition \ref{prop:ex_embed} we obtain that $\Sigma$ admits an 
  immersion $F:W \subset \rr \times M \rightarrow \Sigma \subset N$ with 
  locally embedded
  image, and initial values
  $\restr{F} = \varphi$ and $\restr{\partial_t F} = \nu \circ \varphi$
  in a neighborhood of the initial submanifold $\Sigma_0$.

  Proposition \ref{prop:param_uni} now yields that there is a local 
  diffeomorphism
  $\Psi$ such that $F \circ \Psi^{-1}$ and $F_0$ coincide.
  Hence, the desired result follows.
\hfill \opensquare %

\ack
The author would like  to thank Gerhard Huisken for interesting discussions
and encouragement, 
and Paul Allen for helpful assistance.

\section*{References}

\bibliographystyle{unsrt}

\addcontentsline{toc}{section}{References}

\bibliography{notes} 

\end{document}